\documentclass{article}

\usepackage{latexsym, amsthm}
\usepackage{amsfonts, amsmath, amssymb}
\usepackage{euscript,mathrsfs}
\usepackage{url}
\usepackage{array}
\usepackage[a4paper]{geometry}
\usepackage{graphics}
\usepackage[usenames]{color}
\usepackage[all]{xy}
\usepackage{graphicx}
\usepackage{boxedminipage}

\newtheorem{theorem}{Theorem}[section]

\newtheorem{corollary}[theorem]{Corollary}

\newtheorem{remark}[theorem]{Remark}
\newtheorem{example}[theorem]{Example}

\newtheorem{lemma}[theorem]{Lemma}

\newcommand{\fq}{\mathbb{F}_q}

\newcommand{\RM}{\mathrm{RM}_q}

\newcommand{\red}{\mathrm{red}}
\newcommand{\Ev}{\mathrm{Ev}}

\def\AA{{\mathbb A}}

\newcommand{\Ze}{\mathsf{Z}}

\begin{document}

\title{A note on the generalized Hamming weights of Reed--Muller codes}
\author{Peter Beelen\footnote{Department of Applied Mathematics and Computer Science, Technical University of Denmark, DK-2800, Kongens Lyngby, Denmark, pabe@dtu.dk}}
\date{\empty}
\maketitle

\abstract{
In this note, we give a very simple description of the generalized Hamming weights of Reed--Muller codes. For this purpose, we generalize the well-known Macaulay representation of a nonnegative integer and state some of its basic properties. }

\section{Preliminaries}
\label{sec:in}

Let $\fq$ be the finite field with $q$ elements and denote by $\AA^m:=\AA^m(\fq)$ the $m$-dimensional affine space defined over $\fq$. This space consists of $q^m$ points $(a_1,\dots,a_m)$ with $a_1,\dots,a_m \in \fq$. Let $T(m):=\fq[x_1, \dots , x_m]$ denote the ring of polynomials in $m$ variables and coefficients in $\fq.$ Further let $T_{\le d}(m)$ be the set of polynomials in $T(m)$ of total degree at most $d$. A monomial $X_1^{\alpha_1}\cdots X_m^{\alpha_m}$ is called reduced if $(\alpha_1,\dots,\alpha_m) \in \{0,1,\dots,q-1\}^m$. Similarly a polynomial $f \in T(m)$ is called reduced if it is an $\fq$-linear combination of reduced monomials. We denote the set of reduced polynomials by $T^\red(m)$ and define $T^\red_{\le d}(m):=^T_{\le d}(m) \cap T^{\red}(m)$.

One reason for considering reduced polynomials comes from coding theory. Indeed Reed--Muller codes are obtained by evaluating certain polynomials in the points of $\AA^m$, but the evaluation map
$$\Ev: T(m) \to \fq^{q^m}, \ \makebox{defined by} \ \Ev(f)=(f(P))_{P \in \AA}$$ is not injective. However, its restriction to $T^\red(m)$ is.
In fact the kernel of $\Ev$ consists precisely of the ideal $I \subset T(m)$ generated by the polynomials $x_i^q-x_i$ ($1 \le i \le m$). Working with reduced polynomials is simply a convenient way to take this into account, since for two reduced polynomials $f_1,f_2 \in T(m)$ the equality $f_1+I=f_2+I$ holds if and only if $f_1=f_2$.

The Reed--Muller code $\RM(d,m)$ is the set of vectors from $\fq^{q^m}$ obtained by evaluating polynomials of total degree up to $d$ in the $q^m$ points of $\AA^m$, that is to say:
$$\RM(d,m):=\{(f(P))_{P \in \AA^m} \, : \, f \in T_{\le d}(m)\}.$$
By the above, we also have $\RM(d,m):=\{(f(P))_{P \in \AA^m} \, : \, f \in T_{\le d}^\red(m)\}$ and moreover, we have
\begin{equation}\label{eq:dimRM0}
\dim \RM(d,m)=\dim T^\red_{\le d}(m).
\end{equation}
Reed--Muller codes $\RM(d,m)$ have been studied extensively for their elegant algebraic properties. Their generalized
Hamming weights $d_r(\RM(d,m))$ have been determined in \cite{HP} by Heijnen and Pellikaan.
For a general linear code $C \subseteq \fq^n$ these are defined as follows:
$$d_r(C):=\min_{D \subseteq C: \dim D=r}|\mathrm{supp}(D)|,$$
where the minimum is taken over all $r$-dimensional $\fq$-linear subspaces $D$ of $C$ and where $\mathrm{supp}(D)$ denotes the support size of $D$, that is to say $$\mathrm{supp}(D):=\#\{i \, : \, \exists \, (c_1,\dots,c_n) \in D, \ c_i \not = 0\}.$$
In case of Reed--Muller codes, there is a direct relation between generalized Hamming weights and the number of common solutions to systems of polynomial equations. Indeed, if $D \subset \RM(d,m)$ is spanned by $(f_i(P))_{P \in \AA}$ for $f_1,\dots,f_r \in T_{\le d}^\red(m)$, then $\mathrm{supp}(D)=q^m-\#\Ze(f_1,\dots,f_r)$ where $\Ze(f_1, \dots , f_r):=\{P \in \AA^m \,:\, f_1(P)=\cdots =f_r(P)=0\}$ denotes the set of common zeros of $f_1, \dots , f_r$ in the $m$-dimensional affine space $\AA^m$ over $\fq$.
Therefore, if we define
\begin{equation}
\label{erdmA}
\bar e^{\AA}_r(d,m): = \max\left\{ \left| \Ze(f_1, \dots , f_r) \right| \,:\,  f_1, \dots , f_r \in T^\red_{\le d}(m) \; \text{ linearly independent}\right\} ,
\end{equation}
then $d_r(\RM(d,m))=q^m-\bar e^{\AA}_r(d,m).$
Note that $T^\red(m)$ is a vector space over $\fq$ of dimension $q^m$ and that a reduced polynomial has total degree at most $m(q-1)$. Therefore $T^\red(m)=T^\red_{\le m(q-1)}(m)$. This implies in particular that $\RM(d,m)=\fq^{q^m}$ for $d \ge m(q-1)$. Therefore, we will always assume that $d \le m(q-1).$

The result of Heijnen--Pellikaan in \cite{HP} on the value of $d_r(\RM(d,m))$ can now be restated as follows, see for example \cite{BD}.
\begin{equation}
\label{Hrdm}
\bar e^{\AA}_r(d,m)= \sum_{i=1}^m \mu_i q^{m-i},
\end{equation}
where $(\mu_1, \dots, \mu_{m})$ is the $r$-th $m$-tuple in descending lexicographic order among
all $m$-tuples $(\beta_1, \dots , \beta_{m})\in \{0,1,\dots,q-1\}^m$ satisfying $\beta_1+ \cdots + \beta_{m} \le d$.

Following the notation in \cite{HP}, we denote with $\rho_q(d,m)$ the dimension of $\RM(d,m)$. Equation \eqref{eq:dimRM0} implies that $\rho_q(d,m)=\dim(T_{\le d}^\red(m)).$ In particular, we have
\begin{equation}
\rho_q(d,m)=\dim(T_{\le d}(m))=\binom{m+d}{d}, \ \makebox{if $d \le q-1$},
\end{equation}
since $T_{\le d}(m)=T_{\le d}^\red(m)$ if $d<q.$ Here as well as later on we use the convention that $\binom{a}{b}=0$ if $a<b$. In particular we have $\rho_q(d,m)=0$ if $d<0$.
As shown in \cite[\S 5.4]{AK}, for the general case $d \le m(q-1)$, we have
\begin{equation}\label{eq:dimRM}
\rho_q(d,m)=\dim(T^\red_{\le d}(m))=\sum_{i=0}^d \sum_{j=0}^{m} (-1)^j \binom{m}{j} \binom{m-1+i-qj}{m-1}.
\end{equation}

In this note, we will present an easy-to-obtain expression for $\bar e^{\AA}_r(d,m)$ involving a certain representation of the number $\rho_q(d,m)-r$ that we introduce in the next section.

\section{The $d$-th Macaulay representation with respect to $q$}

Let $d$ be a positive integer. The $d$-th Macaulay (or $d$-binomial) representation, of a nonnegative integer $N$ is a way to write $N$ as sum as certain binomial coefficients. To be precise
\begin{equation*}
N=\sum_{i=1}^d \binom{s_i}{i},
\end{equation*}
where the $s_i$ integers satisfying $s_d>s_{d-1}> \cdots > s_1 \ge 0$. The usual convention that $\binom{a}{b}=0$ if $a<b$, is used. For example, the $d$-th Macaulay representation of $0$ is given by $0=\sum_{i=1}^d \binom{i-1}{i}.$ Given $d$ and $N$ the integers $s_i$ exist and are unique. The Macaulay representation is among other things used for the study of Hilbert functions of graded modules, see for example \cite{green}. It is well known (see for example \cite{green}) that if $N$ and $M$ are two nonnegative integers with Macaulay representations given by $(k_d,\dots,k_1)$ and $(\ell_d,\dots,\ell_1)$ then $N \le M$ if and only if $(k_d,\dots,k_1) \preccurlyeq (\ell_d,\dots,\ell_1)$, where $\preccurlyeq$ denotes the lexicographic order.

For our purposes it is more convenient to define $m_i:=s_i-i$. We then obtain
\begin{equation}\label{eq_Macaulay}
N=\sum_{i=1}^d \binom{m_i+i}{i},
\end{equation}
where $m_i$ are integers satisfying $m_d \ge m_{d-1} \ge \cdots \ge m_1 \ge -1.$ The reason for this is that for $d \le q-1$ we have $\rho_q(d,m)=\binom{m+d}{d}$. Therefore, we can interpret Equation \eqref{eq_Macaulay} as a statement concerning dimensions of the Reed--Muller codes $\RM(i,m_i)$. For a suitable choice of $N$, it turns out that the $m_i$ completely determine the value of $\bar e^{\AA}_r(d,m)$ if $d \le q-1$. For $d \ge q$, even though the dimension $\rho_q(d,m)$ is not longer given by $\binom{m+d}{d}$, there exists a variant of the usual $d$-th Macaulay representation that turns out to be equally meaningful for Reed--Muller codes. Before stating this representation, we give a lemma.
\begin{lemma}\label{lem:help}
Let $m \ge 1$ be an integer. We have $$\rho_q(d,m) = \sum_{i=0}^{\min\{d,q-1\}}\rho_q(d-i,m-1).$$
\end{lemma}
\begin{proof}
Any polynomial $f \in T(m)$ can be seen as a polynomial in the variable $X_m$ with coefficients in $T(m-1)$. This implies that $T(m) = \sum_{i \ge 0} X_m^i T(m)$, where the sum is a direct sum. Similarly we can write $$T^\red_{\le d}(m)=\sum_{i=0}^{\min\{d,q-1\}}X_m^iT^\red_{\le d-i}(m-1).$$ The result now follows.
\end{proof}
A consequence of this lemma is the following.
\begin{corollary}\label{cor:help}
Let $d=a(q-1)+b$ for integers $a$ and $b$ satisfying $a \ge 0$ and $1 \le b \le q-1.$ Further suppose that $m \ge a$. Then
$$\rho_q(d,m)-1=\sum_{j=0}^{a-1} \sum_{\ell=0}^{q-2} \rho_q(d-j(q-1)-\ell,m-j-1) + \sum_{i=1}^b\rho_q(i,m-a-1).$$
\end{corollary}
\begin{proof}
This follows using Lemma \ref{lem:help} repeatedly. First applying the lemma to each sum within the double summation on the right-hand side, we see that
\begin{multline*}
\sum_{j=0}^{a-1} \sum_{\ell=0}^{q-2} \rho_q(d-j(q-1)-\ell,m-j-1) = \\
\sum_{j=0}^{a-1} \left( \rho_q(d-j(q-1),m-j)-\rho_q(d-(j+1)(q-1),m-j-1) \right) = \\
\rho_q(d,m)-\rho_q(d-a(q-1),m-a)  = \rho_q(d,m)-\rho_q(b,m-a).
\end{multline*}
Using the same lemma to rewrite the single summation on the right-hand side in Equation \eqref{eq:uniqueness2} we see that if $m>a$
$$\sum_{i=1}^b\rho_q(i,m-a-1)=\rho_q(b,m-a)-\rho_q(0,m-a-1)=\rho_q(b,m-a)-1,$$
while if $m=a$, the single summation equals $0$ and the double summation simplifies to $\rho_q(d,m)-1$.
In either case, we obtain the desired result
\end{proof}
We can now show the following.
\begin{theorem}\label{thm:genrepMac}
Let $N \ge 0$ and $d \ge 1$ be integers and $q$ a prime power. Then there exist uniquely determined integers $m_1,\dots,m_d$ satisfying
\begin{enumerate}
\item $N=\sum_{i=1}^d \rho_q(i,m_i),$
\item $-1 \le m_1 \le \cdots \le m_d,$
\item for all $i$ satisfying $1 \le i \le d-q+1$, either $m_{i+q-1} > m_i$ or $m_{i+q-1}=m_i=-1$.
\end{enumerate}
\end{theorem}
\begin{proof}
We start by showing uniqueness. Suppose that
\begin{equation}\label{eq:uniqueness}
N=\sum_{i=1}^d \rho_q(i,m_i)=\sum_{i=1}^d \rho_q(i,n_i)
\end{equation}
and the integers $n_1,\dots,n_d$ and $m_1,\dots m_d$ satisfy the conditions from the theorem. First of all, if $m_d=-1$ or $n_d=-1$ then $N=0$. Either assumption implies that $(m_d,\dots,m_1)=(-1,\dots,-1)=(n_d,\dots,n_1)$. Indeed $n_i\ge 0$ or $m_i \ge 0$ for some $i$ directly implies that $N>0$. Therefore we from now on assume that $m_d\ge 0$ and $n_d \ge 0$. To arrive at a contradiction, we may assume without loss of generality that $n_d \le m_d-1$.

Define $e$ to be the smallest integer such that $n_e \ge 0$. Equation \eqref{eq:uniqueness} can then be rewritten as
\begin{equation}\label{eq:uniqueness1}
N=\sum_{i=1}^d \rho_q(i,m_i)=\sum_{i=e}^d \rho_q(i,n_i)
\end{equation}
Condition 3 from the theorem implies that $n_{i-q+1}<n_i$ for all $i$ satisfying $e \le i \le d $.
Now write $d-e+1=a(q-1)+b$ for integers $a$ and $b$ satisfying $a \ge 0$ and $1 \le b \le q-1$. With this notation, we obtain that for any $0 \le j \le a-1$ and $0 \le \ell \le q-2$ we have that $$n_{d-j(q-1)-\ell} \le n_d-j \le m_d-j-1.$$

In particular choosing $j=a-1$ and $\ell=0$, this implies that $m_d \ge a+n_{q-1+b} \ge a+1+n_b \ge a$. Using these observations, we obtain from Equation \eqref{eq:uniqueness} that
\begin{equation}\label{eq:uniqueness2}
\rho_q(d,m_d) \le N = \sum_{i=e}^d \rho_q(i,n_i) \le \sum_{j=0}^{a-1} \sum_{\ell=0}^{q-2} \rho_q(d-j(q-1)-\ell,m_d-j-1) + \sum_{i=1}^b\rho_q(e+i-1,m_d-a-1).
\end{equation}
Applying the same technique as in the proof of Corollary \ref{cor:help}, we derive that $$\sum_{j=0}^{a-1} \sum_{\ell=0}^{q-2} \rho_q(d-j(q-1)-\ell,m_d-j-1)=\rho_q(d,m_d)-\rho_q(b+e-1,m_d-a)$$
and Equation \eqref{eq:uniqueness2} can be simplified to
\begin{equation}\label{eq:uniqueness3}
\rho_q(d,m_d) \le \rho_q(d,m_d)-\rho_q(b+e-1,m_d-a) + \sum_{i=1}^b\rho_q(e+i-1,m_d-a-1).
\end{equation}
For $m_d=a$ the right-hand side equals $\rho_q(d,m_d)-1$, leading to a contradiction. If $m_d>q$, Equation \eqref{eq:uniqueness3}
implies
$$
\begin{array}{rcl}
\rho_q(b+e-1,m_d-a) &\le & \sum_{i=1}^b\rho_q(e+i-1,m_d-a-1)\\
\\
 &= & \sum_{j=0}^{b-1}\rho_q(e+b-1-j,m_d-a-1)\\
\\
 &< & \sum_{j=0}^{\min\{e+b-1,q-1\}}\rho_q(e+b-1-j,m_d-a-1)
\\
 & = & \rho_q(b+e-1,m_d-a),
\end{array}
$$
where in the last equality we used Lemma \ref{lem:help}. Again we arrive at a contradiction. This completes the  proof of uniqueness of the $d$-th Macaulay representation with respect to $q$.

Now we show existence. Let $d$, $N$ and $q$ be given. We will proceed with induction on $d$. For $d=1$, note that $\rho_q(1,m)=m+1$ for any $m \ge -1$. Therefore, for a given $N \ge 0$, we can write $N=\rho_q(1,N-1)$.

Now assume the theorem for $d-1$. There exists $m_d \ge -1$ such that
\begin{equation}\label{eq:existence}
\rho_q(d,m_d) \le N < \rho_q(d,m_d+1).
\end{equation}
Applying the induction hypothesis on $N-\rho_q(d,m_d)$, we can find $m_{d-1},\dots,m_1$ satisfying the conditions of the theorem for $d-1$. In particular we have that
\begin{enumerate}
\item $N-\rho_q(d,m_d)=\sum_{i=1}^{d-1} \rho_q(i,m_i),$
\item $-1 \le m_1 \le \cdots \le m_{d-1},$
\item $m_{i+(q-1)} > m_i$ for all $1 \le i \le d-q.$
\end{enumerate}
Clearly this implies that $N=\sum_{i=1}^{d} \rho_q(i,m_i),$ but it is not clear a priori that $m_1,\dots,m_d$ satisfy conditions 2 and 3 as well. Conditions 2 and 3 would follow once we show that $m_d \ge m_{d-1}$ and either $m_d > m_{d-q+1}$ or $m_d=m_{d-q+1}=-1$. First of all, if $m_d=-1$, then $N=0$ and $(m_d,\dots,m_1)=(-1,\dots,-1)$. Hence there is nothing to prove in that case. Assume $m_d \ge 0$. From Equation \eqref{eq:existence} and Lemma \ref{lem:help} we see that
\begin{equation}\label{eq:existence1}
N-\rho_q(d,m_d)<\rho_q(d,m_d+1)-\rho_q(d,m_d)=\sum_{i=1}^{\min\{d,q-1\}}\rho_q(d-i,m_d).
\end{equation}
First suppose that $d \le q-1$. First of all, Condition 3 is empty in that setting. Further, Equation \eqref{eq:existence1} implies
$$N-\rho_q(d,m_d) < \sum_{i=1}^{d}\rho_q(d-i,m_d) = \sum_{i=1}^{d-1}\rho_q(d-i,m_d) +1$$
and hence
$$N-\rho_q(d,m_d) \le \sum_{i=1}^{d-1}\rho_q(d-i,m_d)= \sum_{j=0}^{d-2}\rho_q(d-1-j,m_d) < \rho_q(d-1,m_d+1).$$
This shows that $m_{d-1} \le m_d$ as desired.

Now suppose that $d \ge q$. In this situation Equation \eqref{eq:existence1} implies
$$N-\rho_q(d,m_d) < \sum_{i=1}^{q-1}\rho_q(d-i,m_d) = \sum_{j=0}^{q-2}\rho_q(d-1-j,m_d)< \rho_q(d-1,m_d+1).$$
Hence $m_{d-1} \le m_d$ as before. Finally assume that $m_d \le m_{d-q+1}$. Then by the previous and Condition 2, we have $m_{d}= m_{d-1}=\cdots = m_{d-q+1}$. Hence $N \ge \sum_{i=0}^{q-1}\rho_q(d-i,m_d)=\rho_q(d,m_d+1)$ which is in contradiction with Equation \eqref{eq:existence}. This concludes the induction step and hence the proof of existence.
\end{proof}

We call the representation of $N$ in the above theorem the $d$-th Macaulay representation of $N$ with respect to $q$. One retrieves the usual $d$-th Macaulay representation letting $q$ tend to infinity. We refer to $(m_d,\dots,m_1)$ as the coefficient tuple of this representation. A direct corollary of the above is the following.
\begin{corollary}\label{cor:greedy}
The coefficient tuple $(m_d,\dots,m_1)$ of the $d$-th Macaulay representation with respect to $q$ of a nonnegative integer $N$ can be computed using the following greedy algorithm: The coefficient $m_{d-i}$ can be computed recursively (starting with $i=0$) as the unique integer $m_{d-i} \ge -1$ such that
$$\rho_q(d-i,m_{d-i}) \le N-\sum_{j=d-i+1}^{d}\rho_q(j,m_j) <\rho_q(d-i,m_{d-i}+1).$$
\end{corollary}
\begin{proof}
From the existence-part of the proof of Theorem \ref{thm:genrepMac} it follows directly that the given greedy algorithm finds the desired coefficients.
\end{proof}

A further corollary is the following. As before $\preceq$ denotes the lexicographic order.
\begin{corollary}\label{cor:lexorder}
Suppose the $N$ and $M$ are two nonnegative integers whose respective coefficient tuples are $(n_d,\dots,n_1)$ and $(m_d,\dots,m_1)$. Then
$$N \le M \ \makebox{if and only if} \ (n_d,\dots,n_1) \preceq (m_d,\dots,m_1).$$
\end{corollary}
\begin{proof}
Assume $(n_d,\dots,n_1) \preceq (m_d,\dots,m_1).$ It is enough to show the corollary in case $n_d < m_d$. We know from the previous corollary that $n_d$ and $m_d$ may be determined using the given greedy algorithm. In particular this implies that $n_d < m_d$ implies
$$N < \rho_q(d,n_d+1) \le \rho_q(d,m_d) \le M.$$

Assume that $N \le M$. We use induction on $d$. The induction basis is trivial: If $d=1$, then $m_1=M-1$ and $n_1=N-1$. For the induction step, note that $N \le M < \rho_q(d,m_d+1)$ implies by the greedy algorithm that $n_d \le m_d$. If $n_d < m_d$, we are done. If $n_d=m_d$, we replace $N$ with $N-\rho_q(d,m_d)$ and $M$ with $M-\rho_q(d,m_d)$  and use the induction hypothesis to conclude that $(n_d,\dots,n_1) \preceq (m_d,\dots,m_1)$.
\end{proof}

\section{A simple expression for $\bar e^{\AA}_r(d,m)$}

We are now ready to state and prove the relation between the Macaulay representation with respect to $q$ and $\bar e^{\AA}_r(d,m)$.
\begin{theorem}\label{thm:main}
For $1 \le r \le \rho_q(d,m)$, let the $d$-th Macaulay representation of $\rho_q(d,m)-r$ with respect to $q$ be given by
$$\rho_q(d,m)-r=\sum_{i=1}^d \rho_q(i,m_i).$$
Denoting the floor function as $\lfloor \cdot \rfloor$, we have
$$\bar e^{\AA}_r(d,m)=\sum_{i=1}^d \lfloor q^{m_i} \rfloor.$$
\end{theorem}
\begin{proof}
We know from Equation \eqref{Hrdm} that we need to show that $$\sum_{i=1}^d \lfloor q^{m_i} \rfloor=\sum_{i=1}^m \mu_i q^{m-i},$$
with $(\mu_1, \dots, \mu_{m})$ is the $r$-th element in descending lexicographic order among all $m$-tuples $(\beta_1,\dots,\beta_m)$ in $\{0,1,\dots,q-1\}^m$ satisfying
$\beta_1+ \cdots + \beta_{m} \le d$.
First of all note that since $r \ge 1$, we have $\rho_q(d,m)-r < \rho_q(d,m)$. In particular this implies that $m_d \le m-1$.
Therefore the coefficients of the $d$-tuple $(m_d,\dots,m_1)$ are in $\{-1,0,\dots,m-1\}$.
Now for $1 \le i \le m+1$ define $\mu_i:=|\{j \, : \, m_j=m-i\}|.$
Since the $d$-tuple $(m_d,\dots,m_1)$ is nonincreasing by Condition 2 from Theorem \ref{thm:genrepMac}, we can reconstruct it uniquely from the $(m+1)$-tuple $(\mu_{1},\mu_2,\dots,\mu_{m+1}).$
Moreover, Condition 3 from Theorem\ref{thm:genrepMac}, implies that $(\mu_1,\dots,\mu_m) \in \{0,1,\dots,q-1\}^m$, but note that $\mu_{m+1}$ could be strictly larger than $q-1$.
Further by construction we have $\mu_1+\cdots+\mu_m+\mu_{m+1}=d$, implying that $\mu_1+\cdots+\mu_m \le d$.
Note that $\mu_{m+1}$ is determined uniquely by $(\mu_1,\dots,\mu_m)$, since $\mu_0=d-\mu_1-\cdots-\mu_m$.
Therefore the correspondence between the $d$-tuples $(m_d,\dots,m_1)$ of coefficients of the $d$-th Macaulay representations with respect to $q$ of integers $0 \le N < \rho_q(d,m)$
and the $m$-tuples $(\mu_1, \dots , \mu_{m})\in \{0,1,\dots,q-1\}^m$ satisfying $\mu_1+ \cdots + \mu_{m} \le d$, is a bijection.
Moreover by construction we have $$\sum_{i=1}^d\lfloor q^{m_i} \rfloor =\sum_{j=1}^{m+1} \mu_j \lfloor q^{m-j} \rfloor = \sum_{j=1}^{m} \mu_j q^{m-j}.$$

What remains to be shown is that the constructed $m$-tuple coming from the integer $\rho_q(d,m)-r$ is in fact the $r$-th in descending lexicographic order.
First of all, by Corollary \ref{cor:help} we see that for $r=1$ and $d=aq+b$ that the $m$-tuple associated to $\rho_q(d,m)-1$ equals $(q-1,\dots,q-1,b,0,\dots,0)$,
which under the lexicographic order is the maximal $m$-tuple among all $m$-tuples $(\beta_1, \dots , \beta_{m})\in \{0,1,\dots,q-1\}^m$ satisfying $\beta_1+ \cdots + \beta_{m} \le d$.
Next we show that the conversion between $d$-tuples $(m_d,\dots,m_1)$ to $m$-tuples $(\mu_1,\dots,\mu_m)$ preserves the lexicographic order.
Suppose therefore that $1\le r \le s \le  \rho_q(d,m)$. We write $N:=\rho_q(d,m)-s$ and $M:=\rho_q(d,m)-r.$ and denote their Macaulay coefficient tuples with $(n_d,\dots,n_1)$ and $(m_d,\dots,m_1)$.
Since $N \le M$, Corollary \ref{cor:lexorder} implies that $(n_d,\dots,n_1) \preceq (m_d,\dots,m_1)$.
Also, since these $d$-tuples are nonincreasing, this implies that their associated $m$-tuples $(\nu_1,\dots,\nu_m)$ and $(\mu_1,\dots,\mu_m)$ satisfy
$(\nu_1,\dots,\nu_m) \preceq (\mu_1,\dots,\mu_m)$.
Indeed assuming without loss of generality that $\nu_1<\mu_1$ we see that $m_i=n_i=m-1$ for $d-\nu_1 \le i \le d$ but $n_i < m_i=m-1$ for $i=\nu_1+1$.
Now the desired result follows immediately.
\end{proof}

Combining this theorem with the greedy algorithm in Corollary \ref{cor:greedy}, it is very simple to compute values of $\bar e^{\AA}_r(d,m)$ or equivalently of $d_r(\RM(d,m))$. We illustrate this in the two following examples. The parameters in these example also occur in examples from \cite{HP}.
\begin{example}
Let $q=4$, $r=8$, $d=m=3$. Since $d\le q-1$, we may work with the usual Macaulay representation when applying Theorem \ref{thm:main}. We have $\rho_q(d,m)=\binom{6}{3}=20$ and hence
$$\rho_q(d,m)-r=12=\binom{5}{3}+\binom{2}{2}+\binom{1}{1}=\rho_4(3,2)+\rho_4(2,0)+\rho_4(1,0)$$
is the $3$-rd Macaulay representation of $12$. Theorem \ref{thm:main} implies that $\bar e^{\AA}_8(3,3)=4^2+4^0+4^0=18$ and hence $d_8(\mathrm{RM}_4(3,3))=64-18=46$ in accordance with Example 6.10 in \cite{HP}.
\end{example}

\begin{example}
Let $q=2$, $r=10$, $d=3$ and $m=5$. We have $\rho_2(3,5)=26$ by Equation \eqref{eq:dimRM} and hence applying the greedy algorithm from Corollary \ref{cor:greedy}, we compute that
$$\rho_q(d,m)-r=16=15+1+0=\rho_2(3,4)+\rho_2(2,0)+\rho_2(1,-1)$$
is the $3$rd Macaulay representation of $16$ with respect to $2$. Theorem \ref{thm:main} implies that
$\bar e^{\AA}_{10}(3,3)=2^4+2^0=17$ and hence $d_8(\mathrm{RM}_2(3,5))=32-17=15$ in accordance with Example 6.12 in \cite{HP}.
\end{example}

\begin{remark} Theorem \ref{thm:main} is somewhat similar in spirit as Theorem 6.8 from \cite{HP} in the sense that in both theorems a certain representation in terms of dimensions of Reed--Muller codes is used to give an expression for $d_r(\RM(d,m))$. Where we studied decompositions of $\rho_q(d,m)-r$, in \cite{HP} the focus was on $r$ itself. This suggest there may exist a duality between the two approaches, but the similarities seem to stop there. The representation in \cite{HP} is not the Macaulay representation with respect to $q$ that we have used here. For us it is for example very important that each degree $i$ between $1$ and $d$ occurs once in Theorem \ref{thm:genrepMac} (implying that the greedy algorithm terminates after at most $d$ iterations), while this is not the case in Theorem 6.8 \cite{HP}. It could be interesting future work to determine if a deeper lying relationship between the two approaches exists.
\end{remark}


\begin{thebibliography}{22}
\bibitem{AK} E.F.~Assmus Jr. and J.D.~Key, Designs and their Codes, Cambridge University Press, 1992.

\bibitem{BD} P.~Beelen and M.~Datta, Generalized Hamming weights of affine Cartesian codes, Finite Fields and Applications 51, 130--145, 2018.

\bibitem{green} M.~Green, Restrictions of linear series to hyperplanes, and some results of Macaulay and Gotzmann, In: Algebraic Curves and Projective Geometry, Lecture Notes in Mathematics 1389, 76--86, 2006.

\bibitem{HP} P.~Heijnen and R.~Pellikaan, Generalized {H}amming weights of {$q$}-ary {R}eed-{M}uller codes, IEEE Trans. Inform. Theory 44(1), 181--196, 1998.

\end{thebibliography}
\end{document}